\documentclass[twoside]{article}

\usepackage[accepted]{aistats2021}

\usepackage[round]{natbib}

\usepackage{amsmath}

\usepackage{enumitem}
\usepackage[ansinew]{inputenc}
\usepackage{upgreek}
\usepackage{graphicx} 
\usepackage{textcomp} 
\usepackage{amsfonts} 
\usepackage{amsmath}
\usepackage{amssymb}
\usepackage{yfonts}
\usepackage{mathtools}
\usepackage{mathrsfs}
\usepackage{latexsym}
\usepackage{array}
\usepackage{float}
\usepackage{amsthm} 
\usepackage{calc}
\usepackage{perpage}
\usepackage{textcomp}
\usepackage{verbatim} 
\usepackage{marvosym}
\usepackage{hieroglf}
\usepackage{multirow} %
\usepackage{rotating} %
\usepackage{dsfont} %

\usepackage[english]{babel}
\usepackage[hidelinks]{hyperref} 
\usepackage{url}
\usepackage{color}
\usepackage{tikz}
\usetikzlibrary{shapes,arrows,positioning,calc}
\usepackage{float}
\usepackage{soul}

\newcommand{\normzero}[1]{\ensuremath{\!|\!| #1 | \! |_{0}}}

\newcommand{\normtwo}[1]{\ensuremath{\!|\!| #1 | \! |_{2}}}

\newtheorem{theorem}{Theorem}[section]

\newtheorem{lemma}{Lemma}[section]

\newtheorem{definition}{Definition}[section]

\DeclareMathOperator*{\sign}{sign}

\newcommand{\1}{{\rm 1}\mskip -4,5mu{\rm l} }

\newcommand{\argmin}{\mathop{\mathrm{arg\,min}}}

\newcommand{\seq}[1]{\{1,\dots,#1\}}

\def\E{\mathbb{E}}

\def\R{\mathbb{R}}

\newcommand{\rp}{\mathbb{R}^p}
\newcommand{\rpp}{\mathbb{R}^{p\times p}}
\newcommand{\rnp}{\mathbb{R}^{n\times p}}

\newcommand{\fdr}{\operatorname{FDR}}
\newcommand{\fdp}{\operatorname{FDP}}
\newcommand{\efdp}{\widehat{\fdp}}

\newcommand{\graph}{\mathcal{G}}
\newcommand{\vset}{\mathcal{V}}
\newcommand{\eset}{\mathcal{E}}
\newcommand{\eeset}{\widehat{\mathcal{E}}}

\newcommand{\datamat}{X}

\newcommand{\vecx}{\mathbf{x}}
\newcommand{\invcov}{{\widehat\Sigma}^{-1}}
\newcommand{\rvx}{x}

\newcommand{\covmat}{\Sigma}
\newcommand{\dist}{\mathcal{N}_p}

\newcommand{\thold}{\hat t}
\newcommand{\hypo}{\mathscr{H}_{(i,j)}:\covmat^{-1}_{ij}=0}
\newcommand{\statun}{{\statu}^{\,\circ}}
\newcommand{\est}{\widehat{\cormat}}
\newcommand{\estnull}{\est^{\circ}}

\newcommand{\h}{\mathscr{H}}
\newcommand{\K}{\mathcal{K}}
\newcommand{\setreordered}{{\mathcal{\statmat}}^{\circlearrowleft}}
\newcommand{\setreorderedsize}{n^\circlearrowleft}
\newcommand{\esetMod}{\ensuremath{\eset'}}

\newcommand{\Generaledge}{\ensuremath{\mathcal{S}}}

\newcommand{\cormat}{R}
\newcommand{\nullcormat}{R^{\circ}}
\newcommand{\statmat}{ \widehat W}
\newcommand{\statu}{ \widehat T}

\usepackage{pdfpages}

\begin{document}

\twocolumn[

\aistatstitle{False Discovery Rates in Biological Networks}

\aistatsauthor{Lu Yu\And Tobias Kaufmann \And Johannes Lederer }

\aistatsaddress{ 
Department of Statistical Sciences\\
University of Toronto 
\And  
NORMENT, 
University of Oslo\\  Oslo University Hospital 
\And 
Department of Mathematics\\
Ruhr-University Bochum } ]

\begin{abstract}
The increasing availability of data has generated unprecedented prospects for  network analyses in many biological fields, such as neuroscience (e.g., brain networks), genomics (e.g., gene-gene interaction networks), and ecology (e.g., species interaction  networks).
A powerful statistical framework for estimating such networks is Gaussian graphical models, but standard estimators for the corresponding graphs are prone to large numbers of false discoveries. 
In this paper, we introduce a novel graph estimator based on knockoffs that imitate the partial correlation structures of unconnected nodes. 
We then show that this new estimator provides accurate control of the false discovery rate and yet large power.
\end{abstract}

\section{Introduction}
Biological processes can often be formulated as networks; examples include gene-gene regulation networks~\citep{Emmert2014, Hecker2009}, functional brain  networks~\citep{Bullmore2009}, and microbiome networks~\citep{Kurtz2015}.
A common statistical framework for such networks are Gaussian graphical models~\citep{Lauritzen1996}.
(Undirected) Gaussian graphical models describe the biological data as i.i.d.\@ observations of a random vector~$\vecx:=(\rvx_1,\dots,\rvx_p)^\top$ that follows a multivariate normal distribution $\dist(\mathbf{0}_p,\covmat),$ where $\covmat\in\rpp$ is a symmetric, positive definite matrix. 
The graph~$\graph:=(\vset,\eset)$ with node set~$\vset:=\{1,\dots,p\}$ and edge set~$\eset:=\{(i,j)\in \vset \times \vset:  i\neq j, \,\covmat^{-1}_{ij}:=(\covmat^{-1})_{ij}\neq 0\}$ then captures which pairs of the sample vector's coordinates are  dependent conditionally on all other coordinates:
$\rvx_i$ is conditionally independent of~$\rvx_j$ given all other coordinates of~$\vecx$ if and only if~$(i,j)\in\eset$.
For example, in modeling functional brain networks based on functional Magnetic Resonance Imaging (fMRI), 
$p$ is the number of brain regions under consideration, $\rvx_i$ is the activity in the $i$th region,  
and the edge set~$\eset$ denotes the directly connected pairs of regions.

A number of estimators for the edge set~$\eset$ are known. 
Besides  simplistic correlational approaches, 
popular estimators are neighborhood selection~\citep{Mein2006}, which combines node-wise lasso estimates,
and graphical lasso~\citep{Friedman2008, Yuan2007}, which maximizes an $\ell_1$-penalized log-likelihood. 
These two estimators have been equipped with sharp prediction and estimation guarantees even for high-dimensional settings, where the number of samples is not much larger than the number of nodes~$p$~\citep{Ravikumar2011, Rothman2008,Zhuang2018}.
In contrast to such prediction and estimation results,
what is less well understood for  high-dimensional Gaussian graphical models is inference.

Our objective is inference in terms of  control over the false discovery rate (FDR), which is the expected proportion of falsely selected edges over all selected edges.
Such control can make network estimation more reliable, which is particularly useful in biology as many biological networks seem to be hard to unravel---see~\citep{Zhang2018} for corresponding comments regarding brain imaging, for example.
Formally, the FDR is defined as 
\begin{equation}\label{re:fdr}
    \fdr:=\E[\fdp]\,,
\end{equation}
where 
 \begin{equation}\label{re:fdp}
    \fdp:=
    \frac{\#\bigl\{(i,j): (i,j)\notin\eset \text{ and }(i,j)\in \eeset \, \bigr\}}{\#\bigl\{(i,j):(i,j)\in \eeset \,\bigr\}\vee 1} 
\end{equation}
is the  false discovery proportion for an estimator  that returns the  edge set $\eeset\subset \vset\times \vset$, and $a\vee b:=\max\{a,b\}.$
We say that an estimator controls the FDR at level $q$ if $\fdr\le q.$
In the language of hypothesis testing, FDR control is the adjustment to multiple testing for the  hypotheses $\hypo$ for $i\neq j.$

We establish an estimator based on knock-offs. 
In a regression-type setting, knock-offs are ``fake predictors'' that allow one to approximately count the number of falsely included variables ~\citep{Barber2015,Candes2016,Dai2016}.
The knock-offs are supposed to maintain the original features' correlation structure but to be only weakly correlated  with the original features.
Since the relevant predictors tend to have stronger association with the response than their knock-off counterparts, the  number of falsely included variables can be approximated by comparing the estimated signals of the original predictors and their knock-off counterparts.
In a graphical model setting, we introduce knock-offs 
 as ``fake edges.''
Rather than maintaining correlation structures among the original nodes,
 they mimic partial correlations between separate, conditionally independent pairs of nodes. 
We then compare the signals of the sample partial correlations and their knock-off counterparts.

Our contributions can be summarized as follows.
\begin{enumerate}
    \item We introduce a method for FDR control in graphical modeling, where inferential methods have been scarce:
    we compare the existing popular graph estimation methods and show their limitations on the FDR control.
    \item We further support our method both mathematically and numerically:
    we establish theoretical guarantee for both approximate and exact FDR control in Theorems~\ref{re:thm1} and~\ref{re:thm2}; and in Section~\ref{sec:sim}, we demonstrate the proposed method achieves the FDR control and yields higher power than other popular graph estimation methods.
    \item We apply the proposed method to three biological network data sets, and show in Section~\ref{sec:realdata} that our method provides new insights into biological data.
\end{enumerate}

We provide a free implementation that can be  applied to  networks within and beyond the exemplified domains on \url{https://github.com/LedererLab/GGM-FDR}.

\paragraph{Related literature}
\cite{Drton2004} provides  conservative simultaneous confidence intervals for the elements of the precision matrix~$\Sigma^{-1}$  in Gaussian graphical models.
\cite{Van2004} studies the tail probability of the proportion of false positives via the family-wise error rate to obtain asymptotic FDR control in $n\to\infty.$ 
\cite{Drton2007} uses \cite{Van2004}'s approach in a multiple testing framework about conditional independence to obtain asymptotic FDR control in $n\to\infty$.
\cite{Liu2013} uses a multiple testing framework about conditional independence to obtain asymptotic FDR control in $n,p\to\infty$. 
\cite{Jankova2015} establishes element-wise confidence intervals for~$\Sigma^{-1}$.

\paragraph{Outline of the paper}

The rest of this paper is organized as follows. 
In Section~\ref{sec:motivation},  we demonstrate that new methodology is indeed needed for FDR control in Gaussian graphical models.  
In Section~\ref{sec:method}, we introduce our approach and prove its effectiveness both mathematically and numerically.
In Section~\ref{sec:realdata}, we apply our pipeline to three biological network data sets.
In  Section~\ref{sec:discussion}, we conclude with a discussion.
All the proofs are deferred to the supplement.

\section{Motivation}\label{sec:motivation}

\newcommand{\figuremotivation}[1]{
\begin{figure}[#1]
\centering
\includegraphics[width=63mm]{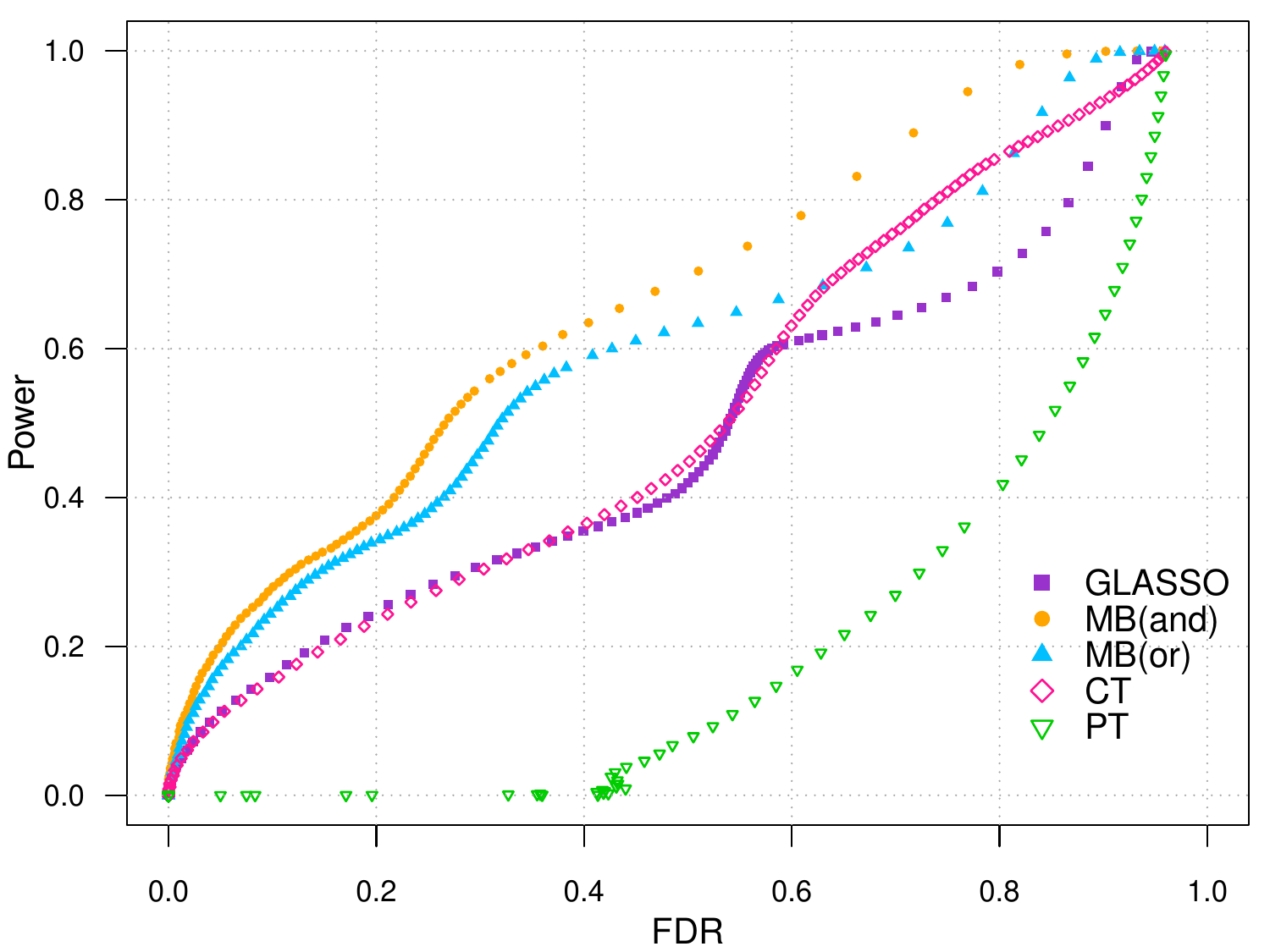}
\vspace{-3mm}
\caption{FDR and power for GLASSO, MB(and),  MB(or), CT, and PT as functions of the tuning parameters.
None of the five methods provides a tuning parameter that leads to both small FDR and large power, and in any case, it is not clear how to calibrate the tuning parameters accordingly in practice.}
\label{fig:motivation}
\end{figure}}

\figuremotivation{h}

We now illustrate numerically why standard methods for estimating Gaussian graphical models do not provide satisfactory FDR control for edge selection. 
Five methods are considered: graphical lasso (GLASSO), neighborhood selection with the ``and-rule'' (MB(and)) and the ``or-rule'' (MB(or)), thresholding the correlation matrix (CT), and thresholding the partial correlation matrix (PT).
The number of nodes is set to $p=400.$ 
The \texttt{huge} package in \texttt{R}~\citep{Zhao2012} is used to generate a covariance matrix~$\covmat$ that commensurates with an undirected band graph model;
in fMRI studies, for example, band graphs reflect that 
connectivities are expected to decrease with increasing spatial distance between the regions~\citep{Bu2017}. 
The condition number of the covariance matrix $\covmat$ is set to $200,$ and the sparsity level is set to $1/25$; 
these settings yield graphs that are diverse and moderately dense. 
Finally, $20$~independent data sets with each one consisting of $n=800$~independent samples from $\dist(0,\covmat)$ are generated.

Using again the \texttt{huge} package, the estimators are computed along a fine grid of tuning parameters.
The estimators' accuracy is evaluated in terms of FDR---see~\eqref{re:fdr}---and in terms of power 
\begin{equation*}
    \text{Power}:= \frac{\# \Bigr\{(i,j): (i,j)\in \eset \text{ and }(i,j)\in \eeset\,\Bigr\}}{\#\Bigr\{(i,j):(i,j)\in \eset \Bigr\}\vee 1}\,,
\end{equation*}
which is the proportion of the number of correctly estimated edges to the total number of edges.

Both FDR and power are averaged over the  20~data sets.

Figure~\ref{fig:motivation} contains the FDR/power-curves along the tuning parameter paths. 
There is not necessarily a tuning parameter that leads to small FDR and large power simultaneously. 
And more importantly, FDR and power can be measured in simulations but not in practice;
this means that \emph{even if} there was a tuning parameter that leads to small FDR and large power, it would be unclear how to find it in practice.
In particular, known calibration schemes such as cross-validation~\citep{Arlot2010}, AIC~\citep{Akaike1974}, BIC~\citep{Schwarz1978}, permutation~\citep{Sabourin2015}, and AV~\citep{Chichignoud2016} are designed for different objectives and are, therefore, not suitable for this task.
Taken together,  standard estimators for Gaussian graphical models do not imply sensible FDR control.

\section{Method}\label{sec:method}

In this section, we introduce our strategy to FDR control  and establish both mathematical and numerical support for its accuracy.
A main ingredient of our strategy are knockoffs that  imitate additional partial correlations.
Accordingly, we refer to our method as ``KO.''

\subsection{The KO Strategy}\label{subsec:method}

The KO strategy consists of three steps: First, we equip the sample partial correlations with knock-off counterparts.
Second, we compare the sample partial correlations and their counterparts through corresponding test statistics. 
Third, we produce estimates based on these test statistics by defining a data-driven threshold.

The three mentioned steps now read in detail:

\textit{Step 1: Constructing knock-offs.}

The starting points of our statistical analysis are the partial correlations. 
The partial correlations give us  direct access to the hypotheses~$\hypo$ via the Hammersley-Clifford theorem~\citep{Grimmett1973}: for any Gaussian random vector~$\vecx=(\rvx_1,\dots,\rvx_p)^\top\sim\mathcal N_p(\mathbf{0}_p,\covmat)$, it holds---see also \cite[Pages 129--130]{Lauritzen1996}---that
\begin{align*}
    \rvx_i \perp \rvx_j |\mathbf{x}_{\vset\setminus \{i,j\}}\, ~\Longleftrightarrow~ \,  \covmat^{-1}_{ij}=0  ~\Longleftrightarrow~ \, \rho_{ij\cdot \vset\setminus\{i,j\}}=0\,,
\end{align*}
where $\rho_{ij\cdot \vset\setminus\{i,j\}}$ denotes the partial correlation between the variables~$\rvx_i$ and~$\rvx_j$ given the remaining $p-2$-dimensional  vector~$\mathbf{x}_{\vset\setminus \{i,j\}}.$

We now use  classical properties of sample correlations and sample partial correlations derived by \cite{Fisher15,Fisher21,Fisher24}.
Consider the data matrix~$\datamat = (\vecx^1,\dots,\vecx^n)^\top\in\rnp,$ where~$\vecx^1,\dots,\vecx^n\in\rp$ are independent and identically distributed  samples from $\mathcal N_p(\mathbf{0}_p,\covmat)$ and assume that $n>p$.
(The latter condition  does not exclude high-dimensional settings in general:  $n\approx p$ \textit{cannot} be approached with classical inferential methods,
and the number of parameters in graphical models is $p(p-1)/2$,
which can be much larger than~$n$ even for $p<n$.).
\cite{Fisher15}~derives the distribution of the sample correlation
\begin{equation*}
  C_{ij}:=\frac{\sum_{l=1}^n(\vecx^l)_i(\vecx^l)_j}{\sqrt{\sum_{l=1}^n((\vecx^l)_i)^2}\sqrt{\sum_{l=1}^n((\vecx^l)_j)^2}}
\end{equation*}
of the coordinates~$i$ and~$j$;
in particular,
that paper yields that if the population correlation is zero,
the statistic
\begin{equation*}
  \frac{C_{ij}}{\sqrt{\bigl(1-(C_{ij})^2)/(n-2)}}
\end{equation*}
follows a Student's t-distribution with $n-2$ degrees of freedom.
\cite{Fisher24} then shows that the corresponding sample \emph{partial} correlation,
which we write as an entry
\begin{align*}
    \cormat_{ij}:=-\frac{(\invcov)_{ij}}{\sqrt{(\invcov)_{ii}(\invcov)_{jj}}}\,,
\end{align*}
 of a matrix~$\cormat\in\R^{p\times p}$  with $\widehat\Sigma:=\datamat^\top\datamat$ the sample covariance matrix of the data matrix~$\datamat$,
has the same distribution as~$C_{ij}$ but with an effective sample size of $n-(p-2)$, where $p-2$ is the number of elements in $\vset\setminus\{i,j\}$.
We can, therefore, conclude that assuming the null-hypothesis~$\hypo$,
then the random variable 
\begin{equation}\label{re:asympT}
Z:=\frac{\cormat_{ij}}{\sqrt{(1-\cormat_{ij}^2)/(n-p)}}
\end{equation}
follows a Student's t-distribution with $n-2-(p-2)=n-p$ degrees of freedom.

Motivated by the above observations, we define the entries of $\nullcormat$ through
\begin{equation}\label{re:nullcormat}
\nullcormat_{ij}:=\nullcormat_{ji}:=
\begin{cases}
    1&\text{~~if~}i=j\\
    \frac{Z_{ij}}{\sqrt{n-p+Z_{ij}^2}}&\text{~~if~}i\neq j
\end{cases},
\end{equation}
where the $Z_{ij}$'s ($i,j\in\{1,\dots,p\}$) are sampled independently from the Student's t-distribution with $n-p$ degrees of freedom. 
These are our knockoff versions of the sample partial correlations:
each element of this matrix mimics sample partial correlations between two conditionally independent nodes.
The diagonal elements of~$\nullcormat$ are set to~1 to equal the diagonal elements of~$\cormat$;
the off-diagonal elements of $\nullcormat$ are in $(-1,1)$.

\textit{Step 2:  Establishing the Test Statistics.}
We now construct the test statistics for the entries of the sample partial correlation matrix~$\cormat$ and its knock-off counterpart~$\nullcormat$.
We first apply elementwise hard-thresholding, which can be written as penalized empirical risk minimization
\begin{equation}\label{re:lasso}
\est(t)\in\argmin_{A\in\mathcal{S}}\bigl\{\,\normtwo{\cormat-A}^2+t^2\,\normzero{A}\bigr\}\,,
\end{equation}
where $t>0$ is the thresholding parameter and $\mathcal{S}$ is the set of symmetric and invertible matrices in $\rpp.$
(Our pipeline also applies to soft-thresholding, which corresponds to the $\ell_0$-term swapped with an $\ell_1$-term, and other estimators, but to avoid digression, we omit the details.) 
The knock-off version of that estimator is 
\begin{equation}\label{re:lasso1}
\estnull(t)\in\argmin_{A\in\mathcal{S}}\bigl\{\,\normtwo{\nullcormat-A}^2+t^2\,\normzero{A}\bigr\}\,.
\end{equation}
We now use those estimators to quantify the signal strengths. 
We define the test statistics matrix~$\statu$ via
\begin{equation}\label{re:edgestats}
    \statu_{ij}:=\sup\Big\{t: \big(\est(t)\big)_{ij}\neq 0\Big\}\,,
\end{equation}
which is the point on the tuning parameter path (ranging from $+\infty$ to~$0$) at which the sample partial correlation between $\rvx_i$ and $\rvx_j$  controlling for other variables first enters the model. 
The test statistic~$\statu_{ij}$ indeed tends to be large if $\cormat_{ij}$ (and, therefore, its underlying population versions $\rho_{ij\cdot \vset\setminus\{i,j\}}$) are large.
Similarly, we can evaluate the signal strength of $\nullcormat_{ij}$ via 
\begin{align}\label{re:edgestatss}
    \statun_{ij}:=\sup\Bigr\{t: \big(\estnull(t)\big)_{ij}\neq 0\Bigr\}\,.
\end{align}
Large values of $\statu_{ij}$ provide evidence against~$\hypo$, while large values of $\statun_{ij}$ provide evidence for $\hypo$;
thus, the larger $\statu_{ij}$ in comparison to $\statun_{ij}$, the more confidently we can reject $\hypo$.

For a detailed assessment of the signal strengths, 
we construct the matrix-valued test statistics  $\statmat\in\rpp$ via
\begin{equation}\label{re:statistics}
\statmat_{ij}:=\statmat_{ji}:=
\begin{cases}
   (\statu_{ij}\vee \statun_{ij})\cdot \sign (\statu_{ij}-\statun_{ij})&\text{~~if~}i\neq j\\
    0&\text{~~if~}i=j
\end{cases}\,.
\end{equation}
The test matrix $\statmat$  depends on $\cormat$ and $\nullcormat$ through $\statu_{ij}$ and $\statun_{ij}$. 
A positive  $\statmat_{ij}$ states that the edge~$(i,j)$ enters the model before its knock-off counterpart;
more generally, the larger $\statmat_{ij}$, 
the more evidence we have against the hypothesis $\hypo$.

\textit{Step 3: Defining a Data-dependent Threshold.}

According to the previous step, large  $\statmat_{ij}$  provide evidence against $\hypo$.
In this step, we quantify this by defining a data-driven threshold~$\thold$ and selecting the edges~$(i,j)$ with $\statmat_{ij}\ge \thold,$ which yields the  estimated edge set $\eeset=\{(i,j)\in\vset\times\vset: \statmat_{ij}\ge \thold\,\}.$
Given a target FDR level~$q$, the threshold is defined as 
\begin{equation}\label{re:threshold}
\thold:=\min\Bigg\{t\in\mathcal{\statmat}:\frac{\#\{ (i,j):\statmat_{ij}\le -t\}}{\#\{ (i,j):\statmat_{ij}\ge t\}\vee 1}\le q\Bigg\}\,,
\end{equation}
where $\mathcal{\statmat}:=\{|\statmat_{ij}|:i,j\in\{1,\dots,p\}\}\setminus\{0\}.$ 
We set $\thold:=\infty$ if the minimum is taken over the empty set.
The minimum is always attained as $\mathcal{\statmat}$ is finite. 

Generally, our thresholding scheme aims at bounding the FDR by bounding an  ``empirical version'' of it.
According to Lemma~\ref{flipsign}  in the supplement, it holds for the statistics matrix $\statmat$ defined in~\eqref{re:statistics}, any edge set~$\eset$ that satisfies $\eset =\esetMod:=\{ (i,j)\in \vset \times \vset:  i\neq j, \,\rvx_i \not\perp \rvx_j\}$, and any threshold $t\geq 0$ that
\begin{align*}
&\#\{(i,j):(i,j)\notin \esetMod , \statmat_{ij}\le -t\}\\
=&_d\#\{(i,j):(i,j)\notin \esetMod, \statmat_{ij}\ge t\}\,,    
\end{align*}
where $=_d$ means equivalence in distribution.
Using this equivalence and that an edge $(i,j)$ is selected if and only if $\statmat_{ij}\ge t,$ we can approximately bound $\operatorname{FDP}(t)$, which we define as the FDP for our pipeline with threshold~$t$ as
\begin{align*}
\fdp(t)
&=\frac{\#\big\{ (i,j):(i,j)\notin \eset, \statmat_{ij}\ge t\big\}}{\#\{ (i,j):\statmat_{ij}\ge t\}\vee 1}\\
&=\frac{\#\big\{ (i,j):(i,j)\notin \esetMod, \statmat_{ij}\ge t\big\}}{\#\{ (i,j):\statmat_{ij}\ge t\}\vee 1}\\
&\approx\frac{\#\big\{ (i,j):(i,j)\notin \esetMod, \statmat_{ij}\le -t\big\}}{\#\{ (i,j):\statmat_{ij}\ge t\}\vee 1}\\
&\le\frac{\#\{ (i,j): \statmat_{ij}\le -t\}}{\#\{ (i,j):\statmat_{ij}\ge t\}\vee 1}
=:\efdp(t)\,.
\end{align*}
We interpret $\efdp(t)$ as an estimate of the FDR.
One can check readily that
\begin{equation*}
\thold=\min\Big\{t\in\mathcal{\statmat}:\efdp(t)\le q\Big\}
\end{equation*}
(and set $\thold:=\infty$ if no such $t$ exists),
which means that our data-driven threshold~$\thold$ controls an empirical version of the FDR.

We show in the next section
that the above scheme provides approximate FDR control. 
If \emph{exact} FDR control is required, 
one can modify the scheme similarly as in  mimic~\cite{Barber2015} by thresholding  more conservatively.
Our corresponding threshold is
\begin{equation}\label{re:threshold+}
\thold_+:=\min\Bigg\{t\in\mathcal{\statmat}:\frac{\#\{ (i,j):\statmat_{ij}\le -t\}+1}{\#\{ (i,j):\statmat_{ij}\ge t\}\vee 1}\le q\Bigg\}\,,
\end{equation}
where again $\mathcal{\statmat}=\{|\statmat_{ij}|:i,j=1,\dots,p\}\setminus\{0\}$ and $\thold_+:=\infty$ if no minimum exists.
The difference to the original threshold~$\thold$ is the additional $+1$ in the numerator, which can make the threshold slightly larger (see Section~\ref{app:intuition} in the supplement  for some intuition).
We call the pipeline of Section~\ref{subsec:method} with $\thold$ replaced by $\thold_+$ the \emph{KO+ scheme}.
In practice, however, we would typically recommend the KO scheme, as it has higher statistical power.

\subsection{Mathematical Support}\label{sec:sim}

We now support our method  mathematically.
We first state the following (all proofs are deferred to the supplementary materials): 
\begin{theorem}[Approximate FDR control]\label{re:thm1}
For any target level $q\in[0,1],$ the KO scheme established in Section~\ref{subsec:method} satisfies
\begin{equation*}
\E \Bigg[\frac{\#\big\{(i,j): (i,j)\notin\esetMod\text{ and }(i,j)\in \widehat{\eset}\,\big\}}{\#\big\{(i,j):(i,j)\in\widehat{\eset}\,\big\}+q^{-1}}\Bigg]\le q\,.
\end{equation*}
\end{theorem}
\noindent  
This bound establishes an FDR-type guarantee.
The left-hand side differs from the FDR in~\eqref{re:fdr} and~\eqref{re:fdp} in two aspects, though: 
First, it contains an additional~$q^{-1}$ in the denominator.
But this difference is negligible unless  the number of selected edges is very small,
and it can even be removed by applying a more conservative threshold (see supplementary materials).
Second, it contains~$\esetMod$ rather than~$\eset$,
 that is, 
 it concerns correlations rather than partial correlations.
 Since~$\esetMod$ can be considerably larger than~$\eset$,
 this means that the theorem cannot guarantee FDR control in general.
 But still,
 it can serve as a first mathematical witness for the potency of our approach.
 
 And this potency is confirmed in simulations indeed.
The simulation setup is the one of Section~\ref{sec:motivation}.
In addition,
the number of samples~$n$ and the number of parameters~$p$ is varied,
and our KO method is evaluated on a fine grid of target FDR levels.
Recall that the setup involves a band graph,
where $\esetMod\gg\eset$.
Hence, in view of the above theory (which does not apply to such cases),
good results in this setup would give a particularly strong argument for our method.

In addition, we can also guarantee exact FDR control for the KO+ scheme: 
\begin{theorem}[Exact FDR Control]\label{re:thm2}
For any target level $q\in[0,1],$ and $\eset=\eset'$, the KO+ scheme satisfies
\begin{equation*}
    \operatorname{FDR}=\E \Bigg[\frac{\#\big\{(i,j): (i,j)\notin\eset \text{ and }(i,j)\in \widehat{\eset}\,\big\}}{\#\big\{(i,j):(i,j)\in\widehat{\eset}\,\big\}\vee 1}\Bigg]\le q\,.
\end{equation*}
\end{theorem}

\subsection{Numerical Support}\label{sec:sim}
We now demonstrate the KO's accuracy numerically. 
We show in particular that it achieves the target FDR levels and has favorable power curves

Note first that KO is easy to implement and fast to compute: 
in particular, it does not require any descent algorithm---similarly as CT and PT but in contrast to GLASSO and MB.

The results are displayed in  Figures~\ref{fig:ttfdr} and~\ref{fig:pwrfdr}. 
In the first figure, the observations are essentially always on or below the diagonal, which  demonstrates that KO provides valid
FDR control.
For GLASSO, MB(or), MB(and), CT, and PT,
in contrast, it is unclear how to calibrate the tuning parameters for such a control.
In the second figure,
the KO-curves are essentially always on or above the curves of the competing methods, which demonstrates that KO provides comparable or more power than the other methods for given FDR level. 
Overall, KO has an attractive FDR-power dependence and achieves the nominal FDR level. 

\vspace{-0.2cm}
\begin{figure}
\centering
\includegraphics[width=60mm]{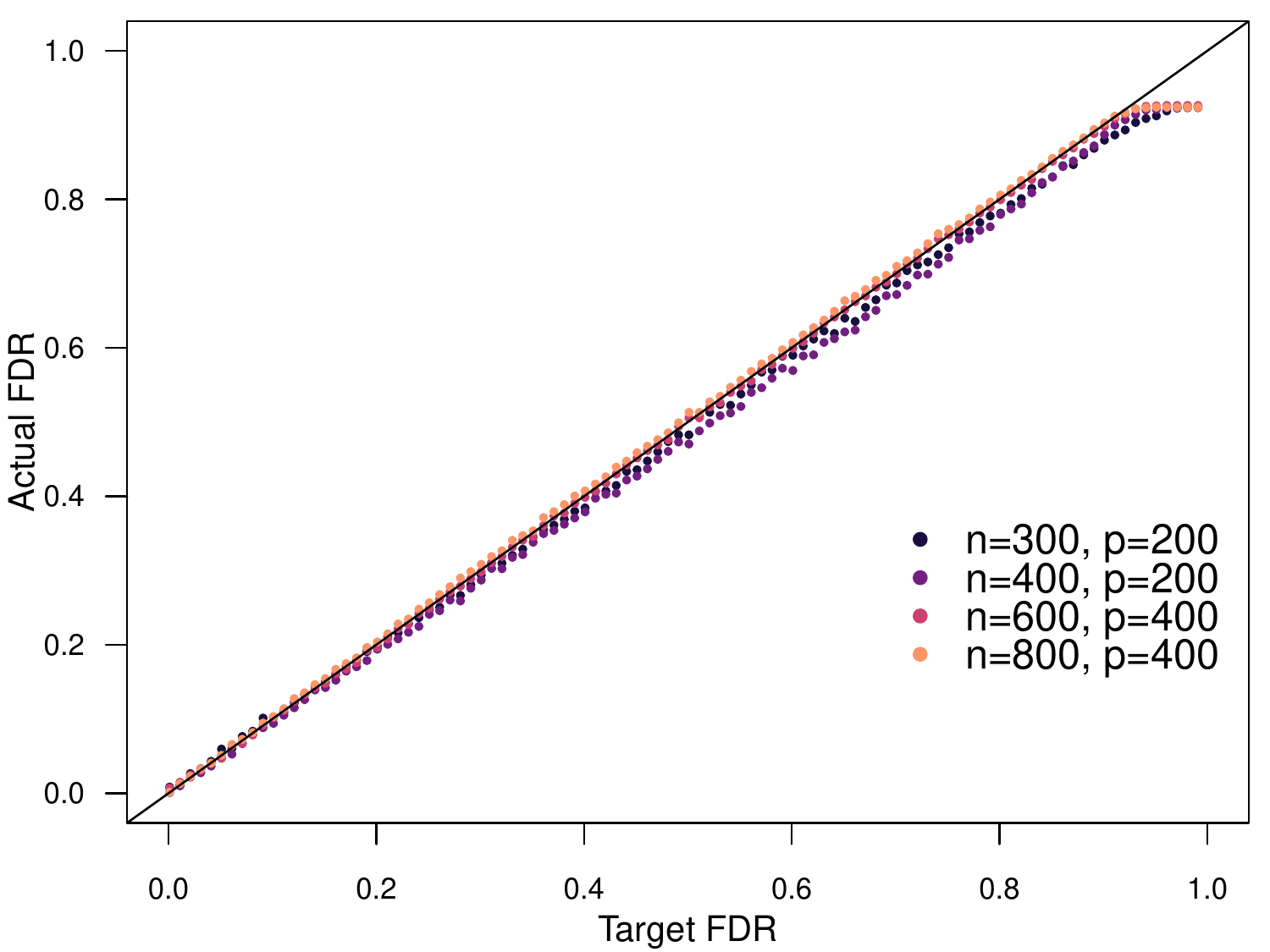}
\vspace{-0.2cm}
\caption{ Actual FDR versus target FDR for KO. 
The curves are basically always on or below the diagonal, meaning that KO provides valid FDR control across all the settings.
}
\label{fig:ttfdr}
\end{figure}

\vspace{-0.2cm}
\begin{figure}[h]
\centering
\includegraphics[width=88mm]{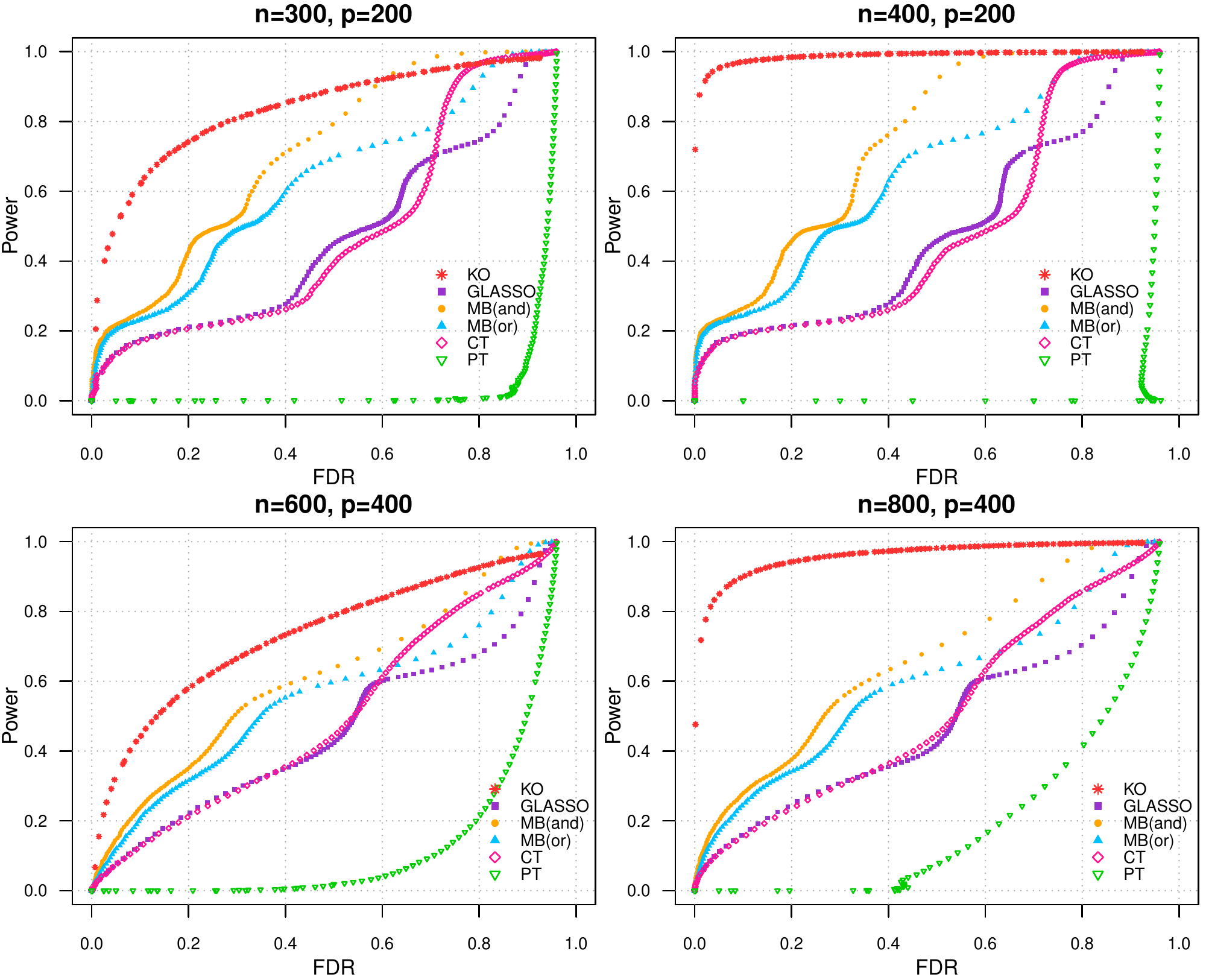}
\vspace{-0.4cm}
\caption{
FDR and power for KO, GLASSO, MB(and), MB(or), CT, and PT as functions  of  the  tuning  parameters.  
Across all settings, 
KO outmatches the other methods in terms of power for given FDR.}
\label{fig:pwrfdr}
\end{figure}

\section{Real Data Analyses }\label{sec:realdata}

We now demonstrate the 
utility of our proposed knock-off method in uncovering biological networks.
We give three examples: brain connectivity networks, microbiological networks in the human gut, and
abundance networks of amphibians. 
The target FDR level is set to~0.2 across all analyses.

\subsection{Brain Connectivity Analysis}

Functional Magnetic Resonance Imaging (fMRI) is a powerful tool to unveil the brain's functional interdependence structures.
The data at hand, described and analyzed in~\cite{Bu2017}, consists of resting-state fMRI acquired at the Department of Neurology at Beijing Hospital from April 2012 through December 2013.
The data set comprises $n = 210$  samples  of the average voxel intensities in $p = 116$ anatomical volumes in $n_{\operatorname{NC}}=10$ individuals with normal cognition.
In line with earlier work~\citep{Horwitz1987,Huang2010},
we restrict our focus to 42 anatomical volumes, further referred to as regions of interest (ROI).
The 42~ROIs are located in the frontal lobe, parietal lobe, occipital lobe, and temporal lobe.

Since we have the data of~$n_{\operatorname{NC}}=10$ subjects, 
we can complement our pipeline with the multiple FDR scheme introduced in~\citep{Xie19} with target FDR level $0.2\times 0.5^{k}$ for the $k$-th individual, $k\in\{1,\dots,10\}.$
We then obtain the continuous graph estimates~$\est_{ij}(\hat t)$ for each individual~$k,$ which is denoted by $\est_{ij}^k(\hat t).$
Then, we calculate the scaled cumulative signal strengths as $\sum_{k\in\text{group}}|\est^k_{ij}(\hat t)|/\max_{l,m}\big\{\sum_{k\in\text{group}}|\est_{lm}^k(\hat t)|\big\}$.

The scaled cumulative signal strengths are displayed in Figure~\ref{fig:42roi}. 
The plot demonstrates that strong connections are predominately between the left and right counterparts of a given region,
which is in line with earlier work on the functional network architecture of the brain~\citep{Honey2035}.

\subsection{Human Microbiome Analysis}\label{sec:HumanMicrobiome}

We now apply the knock-off method to the human microbiome data set of the American Gut Project (\url{http://humanfoodproject.com/americangut/}).
Our specific goal is to learn how the  microbiome is associated with smoking. We use the processed data  that were collected before December, 2018. 
We classify the individuals with smoking frequencies \texttt{Daily}, \texttt{Occasionally(1-2 times/week)}, 
\texttt{Regularly(3-5 times/week)}, and \texttt{Rarely(a few times/month)}  as smokers and the ones with smoking frequency \texttt{Never} as non-smokers.
This yields  $n_{\text{smoker}}=1234$ smokers and $n_{\text{non-smoker}}=15\,640$ non-smokers. 
We incorporate the centered log-ratio transformed~\citep{Aitchison1982}  abundances of the $p=32$ phyla that appear in at least  5\% of the individuals.

To reduce the influence of the imbalanced samples sizes,
we again add the multiple FDR scheme of~\citep{Xie19} to our  method.
Specifically, we uniformly subsample  $n_{\text{sample}}=1234$ individuals from the non-smoker group 10 times.
At each time $k\in\{1,\dots,10\}$, we apply the knock-off method to the corresponding $n_{\text{sample}}\times p$-dimensional data set with  target FDR level $0.2\times 0.5^k$.
Finally, we calculate the scaled cumulative signal strengths, as showed in Figure~\ref{fig:ag}.
The smoker group's data is treated with the vanilla version of our scheme from Section~\ref{subsec:method}. 

We find strong evidence that the graphs of the smokers and non-smokers  differ in their connectivities:
the p-value of 
a corresponding Wilcoxon signed-rank test  is $\ll 10^{-10}$.
In fact---see also the visualization in %
Figure~\ref{fig:AGstats} 
---we find that there are much more interactions in the non-smokers' guts,
which is in agreement with findings in the literature~\citep{Biedermann2013,Savin2018, Stewart2018}.

\subsection{Atlantic Amphibians Abundance  Analysis}\label{sec:AtlanticAmphibians}
We finally analyze abundance data from the Atlantic Forest Biome in South America~\citep{Atlantic2018}.
We specifically consider the $p=30$  most abundant endemic (occurring uniquely in Atlantic Forest) and $p=30$ most abundant non-endemic species  of the order Anura. This ensures that the species appear in at least  0.9\% of the observations.

The corresponding number of study sites for which species abundances are fully documented is $n=346.$
Again, we apply the centered log-ratio transformation to the data.
We find strong evidence for differences in the connectivities of the graphs of the two groups:
the p-value of a corresponding 
Wilcoxon signed-rank test is  $\ll 10^{-5}$.
There are more interactions between the endemic species than between the non-endemic species,
that is, abundances of endemic species are more interconnected among the different species.
See also Figure~\ref{fig:amphibian}, 
which visualizes the scaled connectivity estimates~$|\est_{ij}(\hat t)|/\max_{l,m}\bigl\{|\est_{lm}(\hat t)|\big\}$  from our pipeline.
Since the total number of endemic and non-endemic species is comparable, we hypothesize that this difference is due to a higher level of adaptation of endemic species.
This is in line with with \cite{gorman2014shifts},
which indicates  that  endemic plants have an increased level of adaptation.
However, to the best of our knowledge, our result is the first rigorous quantitative formulation of such a difference between  endemic and non-endemic species.

\section{Discussion}\label{sec:discussion}

We have shown that our  KO pipeline provides effective FDR control and that it can provide new insights into  biological networks.

A topic for further research is the theory:
We provide first theoretical insights in Section~\ref{sec:sim},
but current theorems do yet not establish exact FDR control in general.
However, our numerical results suggest that our theory can be sharpened accordingly.
(In a paper that appeared after ours,
\citep{li2019nodewise} were able to establish a general theory,
but their method is different and computationally much more demanding.)

Another topic for further research are extensions to $p>n$ along the lines of~\citep{Candes2016}.
Our methodology applies very generally otherwise;
in particular, it  applies to arbitrary covariance matrices~$\Sigma$ and asymptotically even to non-Gaussian data.

In summary,
its simplicity and convincing  performance make our pipeline useful for a wide range of applications.

\section{Acknowledgements}

We sincerely thank Jinzhou Li for his insightful comments and the inspiring discussions.

\begin{figure}[h]
\hspace{-0.4cm}
\centering
\includegraphics[width=100mm, height=63mm]{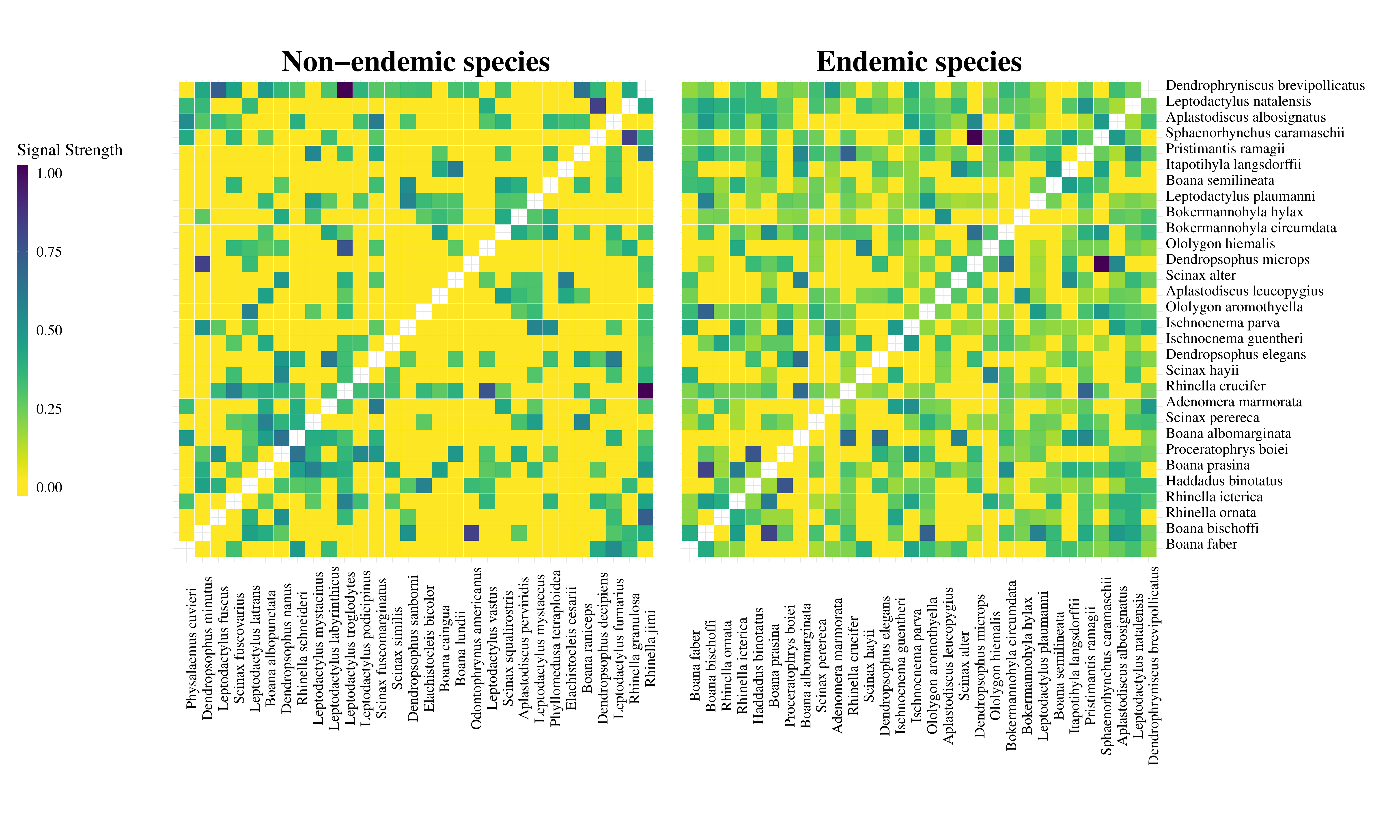}
\vspace{-5mm}
\caption{Signal  strengths for endemic species and non-endemic species in the Atlantic Forest Biome. 
The difference between the two plots in their numbers of gray cells indicates that there are more  connections among endemic species than among non-endemic species.}
\label{fig:amphibian}
\end{figure}

\begin{figure}[h]
\centering
\includegraphics[width=0.5\textwidth]{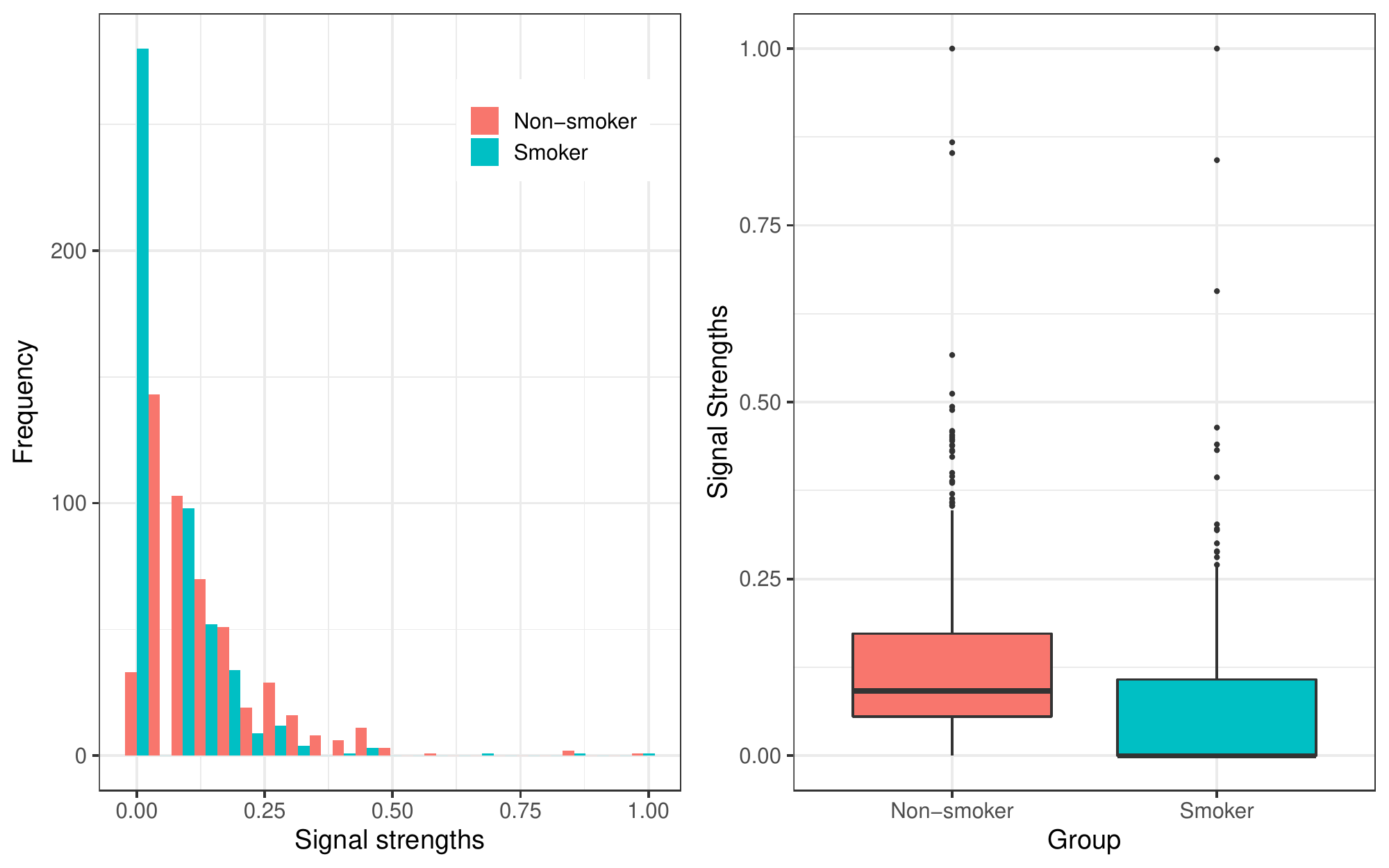}
\caption{Histograms and boxplots of the  signal strengths for the smokers and non-smokers. 
The graphs confirm that the non-smokers' microbiome is more connected than the smokers' microbiome.}
\label{fig:AGstats}
\end{figure}

\vspace{5cm}

\newpage

\onecolumn

\begin{figure}[h]
\centering
\includegraphics[width=120mm]{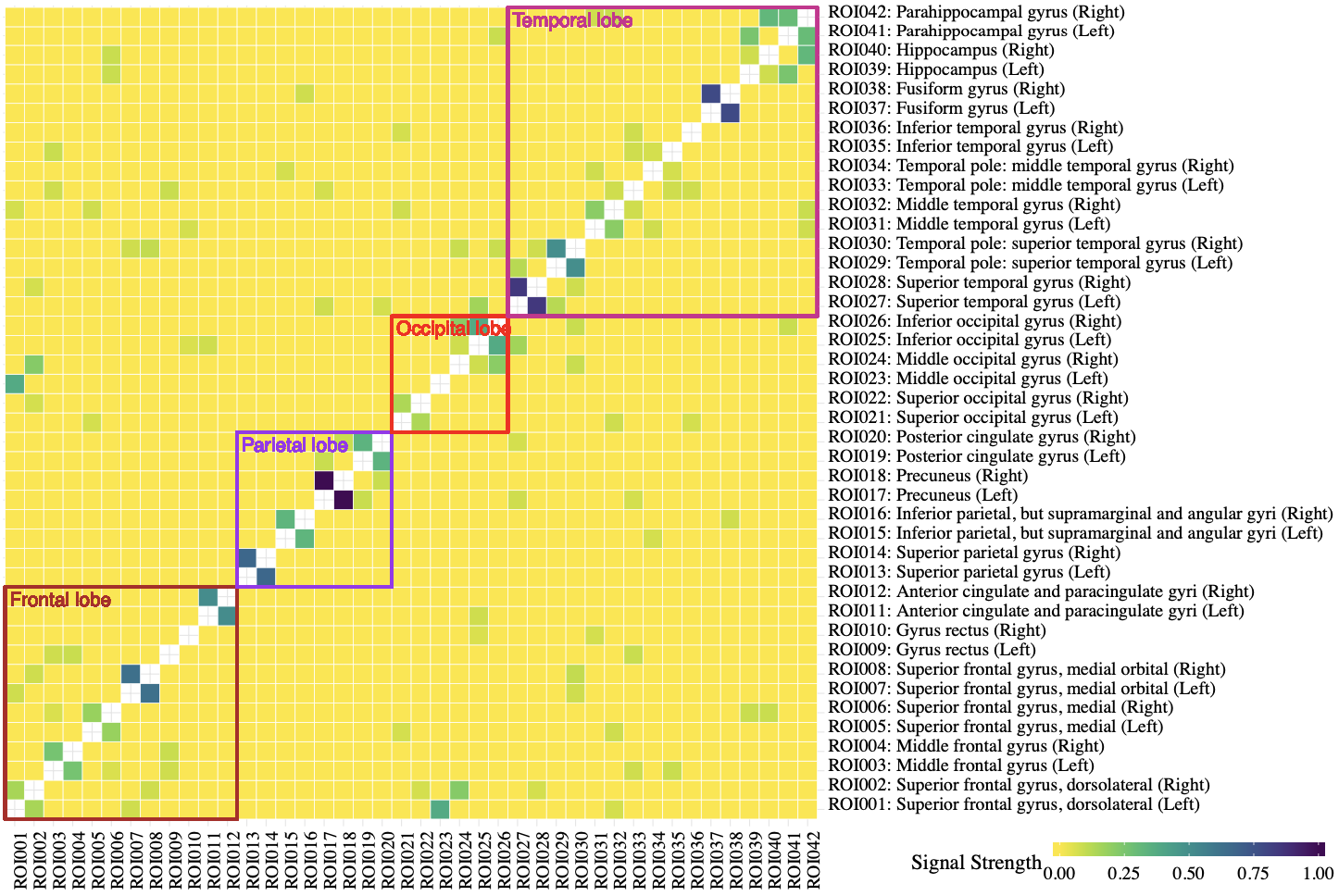}
\vspace{-2mm}
\caption{Cumulative signal strength across $n_{\operatorname{NC}}=10$ individuals for connections among the 42~ROIs.
The four red squares highlight the intra-lobe connections. 
The graph shows that strong connections are most common between  regional counterparts in the left and right hemisphere.}
\label{fig:42roi}
\end{figure}

\vspace{1cm}

\begin{figure}[h]
\centering
\includegraphics[width=125mm,height=77mm]{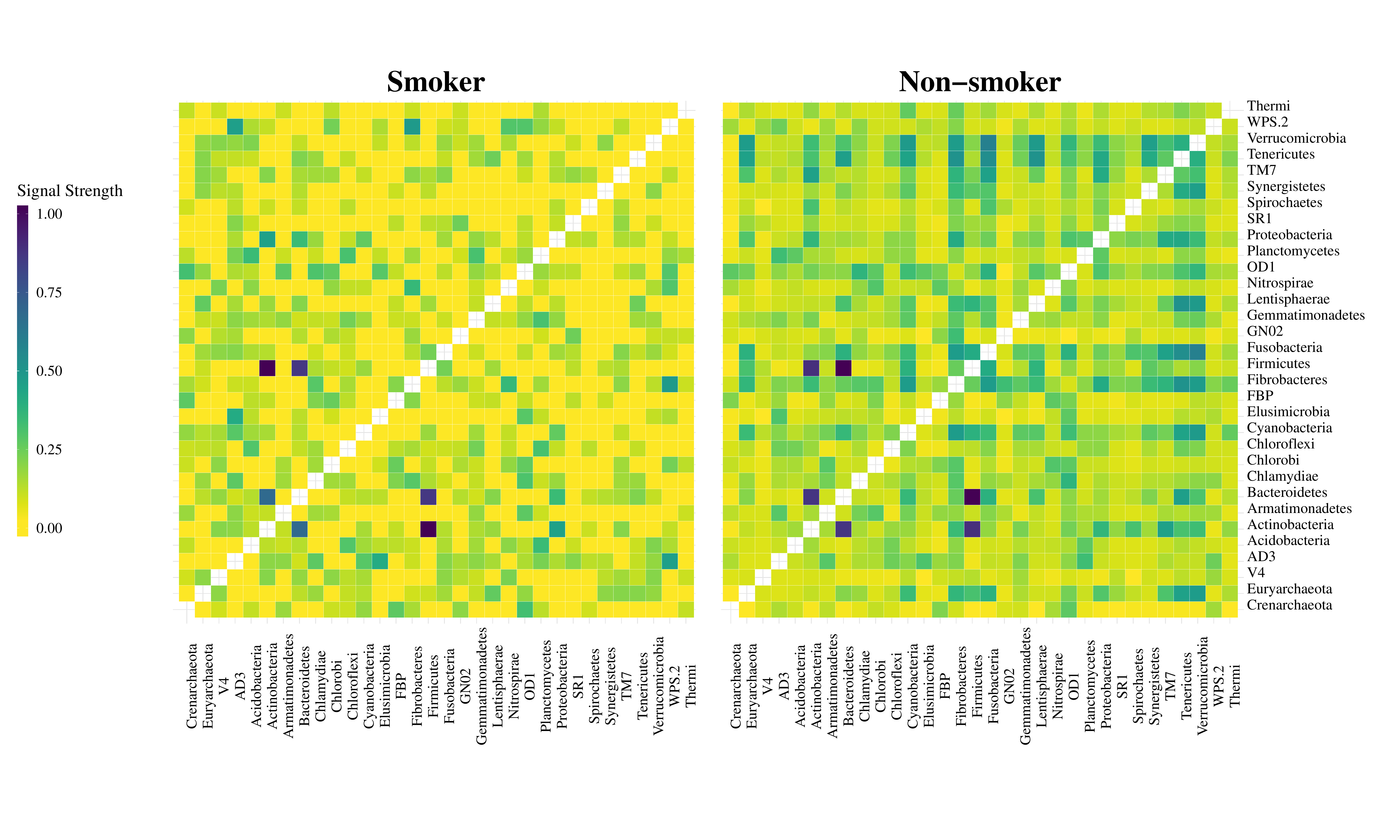}
\vspace{-10mm}
\caption{Cumulative signal  strengths for smoker group and cumulative signal strengths for non-smoker group. 
The graphs show there are more connections among the gut microbiome for non-smokers than for smokers.}
\label{fig:ag}
\end{figure}
\twocolumn

\bibliography{References.bib}
\clearpage
\onecolumn

\appendix\label{app}
\section{Further Intuition}\label{app:intuition}

 To motivate this additional ``+1'' in the KO+ scheme, we consider the FDP for the KO+ pipeline with threshold $\hat t_+$ defined in~\eqref{re:threshold+}:
\begin{align*}
\fdp(\hat t_+)=\,&\,\frac{\#\big\{ (i,j):(i,j)\notin \eset, \statmat_{ij}\ge \hat{t}_+\big\}}{\#\{ (i,j):\statmat_{ij}\ge \hat{t}_+\}\vee 1}\\
\le\, &\, \frac{\#\{ (i,j):(i,j)\notin \eset , \statmat_{ij}\ge \hat{t}_+\}}{1+\#\{ (i,j):(i,j)\notin \eset ,\statmat_{ij}\le -\hat{t}_+\}}\\
&~~~\cdot 
\frac{1+\#\{ (i,j): \statmat_{ij}\le -\hat{t}_+\}}{\#\{ (i,j):\statmat_{ij}\ge \hat{t}_+\}\vee 1}\\
\le\, &\, \frac{\#\{ (i,j):(i,j)\notin \eset , \statmat_{ij}\ge \hat{t}_+\}}{1+\#\{ (i,j):(i,j)\notin \eset ,\statmat_{ij}\le -\hat{t}_+\}}\cdot q\,.
\end{align*}
The first inequality follows from
\begin{equation*}
 \#\{ (i,j):(i,j)\notin \eset , \statmat_{ij}\ge \hat{t}_+\} \le \#\{ (i,j): \statmat_{ij}\ge \hat{t}_+\}\,,
\end{equation*}
and the second inequality follows from the definition of~$\hat{t}_+$.
Using martingale theory, we prove in Appendix~\ref{app:proofs} that 
\begin{align*}
\E\Biggl[\frac{\#\{ (i,j):(i,j)\notin \esetMod, \statmat_{ij}\ge \hat{t}_+\}}{1+\#\{ (i,j):(i,j)\notin \esetMod ,\statmat_{ij}\le -\hat{t}_+\}}\Biggr]\le 1\,.
\end{align*}

\section{Proofs}\label{app:proofs}
The agenda of this section is to establish proofs for Theorems~\ref{re:thm1} and~\ref{re:thm2}.
For this, we define the notion of swapping and study the matrix-valued test statistic~$\statmat\in\rpp.$ 
We write $\statmat$ as $\statmat(\cormat,\nullcormat)$ to emphasize that $\statmat$ is a function of $\cormat$ and $\nullcormat.$

The basis for the proofs is the idea of swapping.
\newcommand{\swapop}{\operatorname{Sub}}
\begin{definition}[Swapping]
Given an edge set $\Generaledge\subset\vset\times\vset$ and a matrix~$M\in\rpp$, we define the substitution operator $\swapop_{\Generaledge,M}:\rpp\to\rpp$ as
\begin{align*}
    A&\mapsto\swapop_{\Generaledge,M}(A):=\begin{cases}
     M_{ij}~~~&\text{if}~(i,j)\in\Generaledge\\
     A_{ij}~~~&\text{if}~(i,j)\notin\Generaledge\,.
    \end{cases}
\end{align*}
We then define the corresponding swapped test matrix as
\begin{equation*}
    \statmat_{\Generaledge}:=\statmat\bigl(\swapop_{\Generaledge,\nullcormat}(\cormat),\swapop_{\Generaledge,\cormat}(\nullcormat)\bigr)\,.
\end{equation*}
\end{definition}
\noindent Given an edge set~$\Generaledge$ and a matrix~$M$, the operator $\swapop_{\Generaledge,M}(A)$ substitutes the elements of~$A$ that have indexes in  $\Generaledge$ by the corresponding elements of~$M$.
Hence,  as compared to the original test matrix~$\statmat$, the new test matrix~$\statmat_{\Generaledge}\equiv\statmat_{\Generaledge}(\cormat,\nullcormat)$ has the entries of~$\cormat$ and~$\nullcormat$ that have indexes in~$\Generaledge$ swapped.
We will see that 
the elements of $\statmat$ and $\statmat_{\Generaledge}$ that correspond to a zero-valued edge have the same distribution, 
while the distributions of other elements can differ.
This gives us leverage for assessing the number of zero-valued edges in a given set~$\Generaledge$.

The swapped test statistics still has an explicit formulation. 
By definition of the original test matrix in~\eqref{re:statistics}, we find
\begin{equation}\label{re:substitution}
 (\statmat_{\Generaledge})_{ij}=
 \begin{cases}
 (\statun_{ij}\vee \statu_{ij})\cdot \sign (\statun_{ij}-\statu_{ij})~&\text{if}~(i,j)\in\Generaledge\\
  (\statu_{ij}\vee \statun_{ij})\cdot \sign (\statu_{ij}-\statun_{ij})~&\text{if}~i\neq j \text{ and }(i,j)\notin\Generaledge\\
  0 ~~~~~~~~~~~~~~~~~~~~~~~~~~~~~~~&\text{if } i = j\,.
\end{cases}
\end{equation}
This means that $\statmat_{\Generaledge}$ is an ``antisymmetric'' version of $\statmat$:
\begin{lemma}[Antisymmetry]\label{antisymmetry}
For every edge set~$\Generaledge\subset\{(k,l)\in\vset\times\vset:k\neq l\}$, it holds that
\begin{align*}
 (\statmat_{\Generaledge})_{ij}=\statmat_{ij}\cdot
 \begin{cases}
 +1& (i,j)\notin\Generaledge \\ 
 -1& (i,j)\in\Generaledge\,.
\end{cases}   
\end{align*}
\end{lemma}
\noindent Hence, swapping  two entries~$\cormat_{ij}, \nullcormat_{ij}$ effects in switching  signs in~$\statmat_{ij}.$

\begin{proof}[Proof of Lemma~\ref{antisymmetry}]
This follows directly from comparing Displays~\eqref{re:statistics} and~\eqref{re:substitution}.
\end{proof}

Now, we show that the coordinates of~$\statmat$ and~$\statmat_{\Generaledge}$ that correspond to a zero-valued edge are equal in distribution.

\begin{lemma}[Exchangeability]\label{exchangeability}
For every zero-valued edge~$(i,j)\in\{(k,l)\in\vset\times\vset:k\neq l,\covmat^{-1}_{kl}=0,x_k\perp x_l\}$, 
it holds that 
\begin{equation*}
(\statmat_{\Generaledge})_{ij}=_d \statmat_{ij}\,,
\end{equation*}
where~$\Generaledge\subset\{(k,l)\in\vset\times\vset:k\neq l\}$ is an arbitrary set of edges and $=_d $ means equality in distribution.
\end{lemma}

\begin{proof}[Proof of Lemma~\ref{exchangeability}]
Our construction of the knock-offs~in \eqref{re:nullcormat} ensures that the sample partial correlation of a zero-valued edge $(i,j)$ and the corresponding knock-off version have the same distribution:
 $\nullcormat_{ij}=_d \cormat_{ij}$.
 This equality in distribution remains true under elementwise thresholding, 
 so that also the corresponding elements of $\statu$ and $\statun$  in~\eqref{re:edgestats} and~\eqref{re:edgestatss}, respectively, are equal in distribution: 
$\statu_{ij}=_d\statun_{ij}$.
This implies that $\sign(\statu_{ij}-\statun_{ij})=_{d}\sign(\statun_{ij}-\statu_{ij})$ (and $\statu_{ij}\vee\statun_{ij}=\statun_{ij}\vee\statu_{ij}$ anyway).
Hence, in view of the definitions of the  test statistics $\statmat$ and $\statmat_{\Generaledge}$   in \eqref{re:statistics} and~\eqref{re:substitution}, respectively,
we find
$(\statmat_{\Generaledge})_{ij}=_{d}\statmat_{ij}$,
as desired.
\end{proof}

We are now ready to discuss the signs of $\statmat_{ij}$.
The below result will be used in the 
proofs of Theorems~\ref{re:thm1} and~\ref{re:thm2}.

\begin{lemma}[Sign-Flip]\label{flipsign}
 For every zero-valued edge~$(i,j)\in\{(k,l)\in\vset\times\vset:k\neq l,\covmat^{-1}_{kl}=0,x_k\perp  x_l\},$
 it holds that
 \begin{equation*}
\statmat_{ij}=_d  -\statmat_{ij}\,.
\end{equation*}
\end{lemma}
\noindent This lemma justifies our previous statement that
\begin{equation*}
    \#\Bigr\{(i,j):(i,j)\notin\esetMod, \statmat_{ij}\le -t\Bigr\}\\
=_d\#\Bigr\{(i,j):(i,j)\notin\esetMod, \statmat_{ij}\ge t\Bigr\}\,.
\end{equation*} 
\begin{proof}[Proof of Lemma~\ref{flipsign}]
Define~\Generaledge\ as the set that only contains the zero-valued edge in question:  $\Generaledge:=\{(i,j)\}$.
 Lemma~\ref{antisymmetry} then yields
\begin{equation*}
    \bigr(\statmat_{\Generaledge}\bigr)_{ij}=\statmat_{ij}\cdot (-1)\,, 
\end{equation*}
while Lemma~\ref{exchangeability} yields
\begin{equation*}
(\statmat_{\Generaledge})_{ij}=_d \statmat_{ij}\,.
\end{equation*}
Combining these two identities concludes the proof.\\
\end{proof}

We now prove Theorems~\ref{re:thm1} and~\ref{re:thm2}.
For this, we start with two sequential hypothesis testing procedures, together with the theoretical results for FDR control. 
Then, we relate these two procedures to KO and KO+ to prove Theorems~\ref{re:thm1} and~\ref{re:thm2}.

We first introduce the two selective sequential hypothesis testing procedures.
Consider a sequence of null hypotheses~$\h_1,\dots,\h_N$ and corresponding ``p-values''~$p_1,\dots,p_N$.
The values $p_1,\dots,p_N$ are not necessarily p-values in a traditional sense, but we will still refer them like that, because they play the same role as p-values here;
in particular, they will need to stochastically dominate a standard uniform random variable, that is, $\Pr (p_l\le u)\le u$ for any $u\in[0,1]$,
which is a typical assumption on traditional p-values---see~\cite[Page 1828]{Ferreira2006}.

We say that a p-value~$p_l$ is a null p-value if the null hypothesis~$\h_l$ is true, and we say $p_l$ is a non-null p-value if $\h_l$ is false with $l\in\{1,\dots,N\}$.

\textit{Selective Sequential Hypothesis Testing~I:}
For the threshold value $1/2$ and any subset $\K\subset\seq{N},$ define
\begin{equation}\label{re:general_thold}
    \hat k := \max \Biggl\{ k\in \K:\frac{\#\bigl\{l\in\seq{k}:p_l> 1/2\bigr\}}{\#\bigl\{l\in\seq{k}:p_l\le 1/2\bigl\}\vee 1}\le  q  \Biggr\}\,.
\end{equation}
Set $\hat k:=0$ if the above set is empty. 
We reject $\h_k$ for all $k\le \hat k$ with $p_k\le 1/2.$
We will see that this procedure achieves the approximate FDR control. Moreover, the KO scheme can be framed as this procedure.

\textit{Selective Sequential Hypothesis Testing~II:}
For the threshold value $1/2$ and any subset $\K\subset\{1,2,\dots, N\},$ define
\begin{equation}\label{re:general_thold+}
    \hat k_+ :=
 \max \Bigg\{ k\in \K:\frac{1+\#\bigl\{l\in\seq{k}:p_l> 1/2\bigr\}}{\#\bigl\{l\in\seq{k}:p_l\le 1/2\bigr\}\vee 1}\le q  \Bigg\}\,.
\end{equation}
Set $\hat k_+:=0$ if the corresponding set is empty.
We reject $\h_k$ for all $k\le \hat k_+$ with $p_k\le 1/2.$
We will also see that this procedure achieves the exact FDR control. Moreover, the KO+ scheme can be cast as this procedure.

Our next result guarantees  FDR control over the Selective  Sequential  Hypothesis  Testing~I and~II.

\begin{lemma}[FDR Control Over the Hypothesis  Testing I  and II]\label{re:seqtesting}
Consider the two selective sequential procedures described above, 
and suppose that the null p-values are i.i.d.,  satisfy $\Pr (p_l\le u)\le u$  for any $u\in[0,1]$, and are independent from the non-null p-values. 
Let $V, V_+$ be the numbers of false discoveries of the two procedures, that is,
\begin{align*}
V:=&\#\bigl\{l\in\seq{\hat k}: p_l\text{~is null and~}p_l\le 1/2\bigr\}\\
V_+:=&\#\bigl\{l\in\seq{\hat k_+}: p_l\text{~is null and~}p_l\le 1/2\bigr\}\,,
\end{align*}
and $R,R_+$ be the total number of discoveries of the two procedures, that is,
\begin{align*}
R:=&\#\bigl\{l\in\seq{\hat k}: p_l\le 1/2\bigr\}\\
R_+:=&\#\bigl\{l\in\seq{\hat k_+}: p_l\le 1/2\bigr\}\,.
\end{align*}
Define $R:=V:=0$ if $\hat k=0,$ and define $R_+:=V_+:=0$ if $\hat k_+=0.$
Then, it holds that
\begin{align*}
\E\bigg[ \frac{V}{R+q^{-1}} \bigg]\le q~~~\text{and}~~~\E\bigg[ \frac{V_+}{R_+\vee 1} \bigg]\le q\,.
\end{align*}
\end{lemma}
\noindent This lemma ensures that the Selective  Sequential  Hypothesis  Testing~I controls a quantity close to the FDR and that the Selective  Sequential  Hypothesis  Testing~II achieves exact FDR control.
These guarantees will  be transferred to the KO and KO+ schemes later by showing that these schemes can be formulated as Selective  Sequential  Hypothesis  Testing~I and~II also.

\begin{proof}[Proof of Lemma~\ref{re:seqtesting}]
We start with the Selective  Sequential  Hypothesis  Testing~I. 
The number of total discoveries is always at least as large as the number of false discoveries: $R\geq V$.
Hence, $R=0$ implies $V=0$, and then it's easy to see that the desired inequalities hold (and are actually equalities).
We can thus assume without loss of generality that $R>0$ in the following.

Using the definition of $V$ as the number of false discoveries, the definition of $R$ as the total number of discoveries, and expanding the fraction, we find
\begin{align*}
&\E\biggl[ \frac{V}{R+q^{-1}} \biggr]\\
=&\E\Biggl[\frac{\#\bigl\{l\in\seq{\hat k}: p_l\text{~is null and~}p_l\le 1/2\bigr\}}{1+\#\bigl\{1\le l\le \hat k: p_l\text{~is null and~}p_l> 1/2\bigr\}} \\
&~~\cdot \frac{1+\#\bigl\{l\in\seq{\hat k}: p_l\text{~is null and~}p_l> 1/2\bigr\}}{R+q^{-1}} \Biggr] \,.
\end{align*}

The number of falsly rejected hypothesis is at most as large as the total number of rejected hypotheses
\begin{equation*}
\#\bigl\{l\in\seq{\hat k}: p_l\text{~is null and~}p_l> 1/2\} \\
\le  \#\bigl\{l\in\seq{\hat k}: p_l> 1/2\bigr\}\,.
\end{equation*}
Moreover, since $R> 0$, the definition of $\hat k$ yields that
\begin{equation*}
 \#\bigl\{l\in\seq{\hat k}: p_l> 1/2\bigr\}
\le q\cdot R\,.
\end{equation*}
Combining these two results gives
\begin{equation*}
\#\bigl\{l\in\seq{\hat k}: p_l\text{~is null and~}p_l> 1/2\bigr\} \le  q\cdot R\,.
\end{equation*}
Plugging this into the previous display and some rearranging provides us with
\begin{align*}
\E\biggl[ \frac{V}{R+q^{-1}} \biggr]&\le \E\Biggl[\frac{\#\bigl\{l\in\seq{\hat k}: p_l\text{~is null and~}p_l\le 1/2\bigr\}}{1+\#\bigl\{l\in\seq{\hat k}:p_l\text{~is null and~} p_l> 1/2\bigr\}} \Biggr]\cdot \frac{1+ q \cdot R}{R+q^{-1}}  \\
&= \E\Biggl[ \frac{\#\bigl\{l\in\seq{\hat k}: p_l\text{~is null and~}p_l\le 1/2\bigr\}}{1+\#\bigl\{l\in\seq{\hat k}: p_l\text{~is null and~}p_l> 1/2\bigr\}}\Biggr] \cdot q\,.
\end{align*}
Inequality~(A.1) of Lemma~1 (martingale process) in the supplement to \cite{Barber2015} gives (set $c=1/2$)
\begin{equation*}
\E\Biggl[ \frac{\#\bigl\{l\in\seq{\hat k}: p_l\text{~is null and~}p_l\le 1/2\bigr\}}{1+\#\bigl\{l\in\seq{\hat k}:p_l\text{~is null and~} p_l> 1/2\bigr\}}\Biggr] \le 1\,.
\end{equation*}
(Here, we have used the assumptions on the p-values.)
Combining this with the previous display gives
\begin{equation*}
\E\bigg[ \frac{V}{R+q^{-1}} \bigg]\le q\,,
\end{equation*}
as desired.

We now prove the FDR control over Selective  Sequential  Hypothesis  Testing~II. 
By definitions of the total discoveries~$V_+$ and false discoveries~$R_+$, it holds that
$V_+=R_+=0$ when $\hat k_+=0$. We then find that
\begin{align*}
\E\bigg[ \frac{V_+}{R_+\vee 1}\cdot \1(0=\hat k_+)\bigg]=0\,,
\end{align*}
which implies
\begin{align*}
\E\bigg[ \frac{V_+}{R_+\vee 1}\bigg]=\E\bigg[ \frac{V_+}{R_+\vee 1}\cdot \1(0<\hat k_+)\bigg]\,.
\end{align*}

Using the definitions of $V_+$ and $R_+$, and expanding the fraction gives
\begin{align*}
&\E\bigg[ \frac{V_+}{R_+\vee 1}\bigg]\\
=&\E\Biggl[ \frac{\#\bigl\{l\in\seq{\hat k_+}: p_l\text{~is null and~}p_l\le 1/2\bigr\}}{1+\#\bigl\{l\in\seq{\hat k_+}: p_l\text{~is null and~}p_l> 1/2\bigr\}}\\
&\,\times \frac{1+\#\bigl\{l\in\seq{\hat k_+}:p_l\text{~is null and~} p_l> 1/2\bigr\}}{\#\bigl\{ l\in\seq{\hat k_+}: p_l\le 1/2\bigr\}\vee 1}\cdot \1( 0< \hat k_+)\Biggr]\,.
\end{align*}

The number of falsly rejected hypothesis is at most as large as the total number of rejected hypotheses
\begin{align*}
\#\bigl\{l\in\seq{\hat k_+}: p_l\text{~is null and~}p_l> 1/2\bigr\} \le  \#\bigl\{l\in\seq{\hat k_+}: p_l> 1/2\bigr\}\,.
\end{align*}
Moreover, by definition of $\hat k_+,$ it holds for $0<\hat k_+$ that
\begin{align*}
\frac{1+\#\bigl\{l\in\seq{\hat k_+}: p_l> 1/2\bigr\}}{\#\bigl\{ l\in\seq{\hat k_+}: p_l\le 1/2\bigr\}\vee 1}\le q\,.
\end{align*}
Combining these two results gives
\begin{align*}
\frac{1+\#\bigl\{l\in\seq{\hat k_+}: p_l\text{~is null and~}p_l> 1/2\bigr\}}{\#\bigl\{ l\in\seq{\hat k_+}: p_l\le 1/2\bigr\}\vee 1}\le  q\,.   
\end{align*}
Plugging this into previous display and some rearranging yields
\begin{align*}
\E\biggl[ \frac{V_+}{R_+\vee 1}\biggr]&\le \E\Biggl[ \frac{\#\bigl\{ l\in\seq{\hat k_+}: p_l\text{~is null and~}p_l\le 1/2\bigr\}}{1+\#\bigl\{l\in\seq{\hat k_+}: p_l\text{~is null and~}p_l> 1/2\bigr\}}\Biggr] \cdot q\,.
\end{align*}
Invoking Inequality of Lemma 1 (martingale process) in the supplement to \cite{Barber2015} again (set $c=1/2$), we find 
\begin{align*}
\E\Biggl[ \frac{\#\bigl\{l\in\seq{\hat k_+}: p_l\text{~is null and~}p_l\le 1/2\bigr\}}{1+\#\bigl\{l\in\seq{\hat k_+}: p_l\text{~is null and~}p_l> 1/2\bigr\}}\Biggr] \le 1\,.
\end{align*}
Combining this with the previous display gives
\begin{align*}
\E\biggl[ \frac{V_+}{R_+\vee 1}\biggr]\le q \end{align*}
as desired.
\end{proof}

We now show that the KO procedure is equivalent to the Selective Sequential Hypothesis Testing~I, and KO+ procedure can be framed as the Selective Sequential Hypothesis Testing~II.  
Then, the desired FDR control over KO and KO+ schemes follows directly from Lemma~\ref{re:seqtesting}.

\begin{proof}[Proof of Theorem~\ref{re:thm1} and Theorem~\ref{re:thm2}]
The proof has two steps:
First, we arrange  the elements of the matrix-valued statistics~$\statmat$ in decreasing absolute value and  define ``p-values'' for each null hypothesis~$\hypo$ based on the corresponding~$\statmat_{ij}.$ 
Second, we connect Selective Sequential Hypothesis Testing~I and the KO scheme as well as Selective Sequential Hypothesis Testing~II and the KO+ scheme and then apply Lemma~\ref{re:seqtesting}.

Define a set of index pairs by~$\setreordered:=\Bigl\{\statmat_{ij}:(i,j)\in\vset\times\vset, \statmat_{ij}\neq 0\Bigr\}$ and denote the cardinality of this set by~$\setreorderedsize:=\text{card}(\setreordered).$
Refer to the elements in $\setreordered$ by $\statmat^1,\dots,\statmat^{\setreorderedsize}$ in a non-increasing order (all elements are non-zero by definition of $\setreordered$):
\begin{equation*}
    |\statmat^1|\ge \cdots\ge |\statmat^{\setreorderedsize}|>0\,.
\end{equation*}

Define the set of indices $\K:=\Bigl\{k\in\{1,\dots,\setreorderedsize-1\big\}: |\statmat^k|>|\statmat^{k+1}|\Bigr\}\bigcup \{\setreorderedsize\}.$ We notice that $\K$ is the index set of unique non-zero values attained by $|\statmat^l|,l\in\{1,\dots,\setreorderedsize\}.$

Define corresponding p-values $p_l$, where $l\in\{1,\dots,\setreorderedsize\}$, based on the test statistic $\statmat^l$:
\begin{align*}
    p_l:=\begin{cases}
    \frac{1}{2}~~~~\statmat^l >0\\
     1~~~~\statmat^l<0\,.
    \end{cases}
\end{align*}
By Lemma~\ref{flipsign} (sign-flip),  $\statmat_{ij}$ is positive and negative equally likely for all zero-valued edges $(i,j)\in\{(k,l)\in\vset\times\vset:k\neq l, \covmat^{-1}_{kl}=0\},$ that is, 
\begin{equation*}
\Pr\Bigl(\statmat_{ij}>0\Bigr)=\Pr\Bigl(\statmat_{ij}<0\Bigr)=\frac{1}{2}\,.
\end{equation*}
Combining this with the definition of the p-value~$p_l,$
it holds that for any null p-value $p_l$ that
\begin{equation*}
\Pr\biggl(p_l=\frac{1}{2}\biggr)=\Pr (p_l=1)=\frac{1}{2}\,,
\end{equation*}
which implies $\Pr(p_l\le u)\le u$ for all $u\in[0,1].$ 
By Lemma~\ref{flipsign},  we find the null p-values are i.i.d., satisfy $\Pr(p_l \le u) \le u$ for any $u\in [0, 1]$, and are independent from the non-null p-values.
By definition of the  p-value~$p_l$,  it holds for any $k\in \K$ that 
\begin{equation*}
\#\bigl\{l\in\seq{k}: p_l>1/2\bigr\} = \#\Bigl\{ l\in\seq{k}:\statmat^l<0\Bigr\}\,.
\end{equation*}
Due to the assumed ordering  $|\statmat^1|\geq \dots\geq|\statmat^{\setreorderedsize}|>0$, we have
\begin{equation*}
   - |\statmat^1|\le \cdots\le - |\statmat^{\setreorderedsize}|<0\,.
\end{equation*}
So, it holds for any $\statmat^l<0$ that 
\begin{equation*}
  - |\statmat^1|\le\cdots\le-|\statmat^{l-1}|\le\statmat^l\le -|\statmat^{l+1}|\le\cdots\le -|\statmat^{\setreorderedsize}|\,,
\end{equation*}
which implies
\begin{equation*}
\#\Bigl\{l\in\seq{k}:\statmat^l<0\Bigr\}=\#\Bigl\{ l\in\seq{\setreorderedsize}:\statmat^l\le -|\statmat^k|\Bigr\} \,.
\end{equation*}
Combining this with the previous display yields
\begin{equation}\label{re:equiv1}
\#\bigl\{l\in\seq{k}: p_l>1/2\bigr\} = \#\Bigl\{ l\in\seq{\setreorderedsize}:\statmat^l\le -|\statmat^k|\Bigr\}\,.
\end{equation}
By the same arguments, we obtain
\begin{equation}\label{re:equiv2}
\#\bigl\{ l\in\seq{k}: p_l\le 1/2\bigr\}= \#\Bigl\{ l\in\seq{\setreorderedsize}:\statmat^l\ge |\statmat^k|\Bigr\}\,.
\end{equation}
Plugging these two displays together, we find
\begin{equation*}
\frac{\#\bigl\{l\in\seq{k}: p_l>1/2\bigr\}}{\#\bigl\{ l\in\seq{k}: p_l\le 1/2\bigr\}\vee 1}=\frac{\#\Bigl\{ l\in\seq{\setreorderedsize}:\statmat^l\le -|\statmat^k|\Bigr\}}{\#\Bigl\{ l\in\seq{\setreorderedsize}:\statmat^l\ge |\statmat^k|\Bigr\}\vee 1}\,.
\end{equation*}

Finding the largest $k\in \K$ such that the ratio on the left-hand side is below $q$ is---in view of the non-increasing ordering of the $|\statmat^k|$'s---equivalent to finding the smallest $|\statmat^k|$ over $k\in \K$ such that the right-hand side is below $q$.
By definition of the threshold value~$\hat k$ of Selective Sequential Hypothesis Testing~I in Display~\eqref{re:general_thold}, 
this means that 
\begin{equation*}
    \hat k = \max\Biggl\{k\in\K: \frac{\#\bigl\{l\in\seq{\setreorderedsize}:\statmat^l\le -|\statmat^k|\bigr\}}{\#\bigl\{ l\in\seq{\setreorderedsize}: \statmat^l\ge |\statmat^k|\bigr\}\vee 1}\le q\Biggr\}\,.
\end{equation*}

Comparing to the definition of the KO threshold in Display~\eqref{re:threshold}, we find that~$\statmat^{\hat k}$ is equal to~$\thold$.
This equality implies that the KO scheme is equivalent to the Selective  Sequential  Hypothesis  Testing~I,
which gives us the desired FDR control. %

Plugging~\eqref{re:equiv1} and~\eqref{re:equiv2} together, it also holds for $k\in\K$ that
\begin{equation*}
\frac{1+\#\bigl\{l\in\seq{k}: p_l>1/2\bigr\}}{\#\bigl\{ l\in\seq{k}: p_l\le 1/2\bigr\}\vee 1}=\frac{1+\#\Bigl\{l\in\seq{\setreorderedsize}:\statmat^l\le -|\statmat^k|\Bigr\}}{\#\Bigl\{ l\in\seq{\setreorderedsize}:\statmat^l\ge |\statmat^k|\Bigr\}\vee 1}\,.    
\end{equation*}
By the definition of the threshold value~$\hat k_+$ of the Selective Sequential Hypothesis Testing~II in Display~\eqref{re:general_thold+}, 
this means that 
\begin{equation*}
    \hat k_+ = \max\Biggl\{k\in\K: \frac{1+\#\bigl\{l\in\seq{\setreorderedsize}:\statmat^l\le -|\statmat^k|\bigr\}}{\#\bigl\{l\in\seq{\setreorderedsize} : \statmat^l\ge |\statmat^k|\bigr\}\vee 1}\le q\Biggr\}\,.
\end{equation*}
Comparing to the definition of the KO+ threshold in Display~\eqref{re:threshold+}, we find that~$\statmat^{\hat k_+}$ is equal to~$\thold_+$.
This equality implies that the KO scheme is equivalent to the Selective  Sequential  Hypothesis  Testing~II.
The desired FDR control of KO+ scheme follows from Lemma~\ref{re:seqtesting}.

\end{proof}

\end{document}